\documentclass[aps,pra,superscriptaddress,showpacs,twocolumn]{revtex4-1}
\usepackage{graphicx}
\usepackage{bm}
\usepackage{amssymb} 
\usepackage{amsmath} 
\usepackage{amsthm}
\usepackage{hyperref}
\usepackage[utf8]{inputenc}
\hypersetup{
     colorlinks=true,      
    linkcolor=blue,        
    citecolor=blue,
    filecolor=magenta, 
    urlcolor=black          
}

\newcommand{\expec}[1]{\langle #1\rangle}
\newcommand{\Tr}{\mathrm{Tr}\,} 

\usepackage{mathbbol}

 \newtheorem{definition}{Definition}
 \newtheorem{theorem}{Theorem}
  \newtheorem{observation}{Observation}
 \newtheorem{lemma}{Lemma}
 \newtheorem{corollary}{Corollary}
  \newtheorem{conjecture}{Conjecture}

\begin{document}

\title{Continuous-variable supraquantum nonlocality} 
\author{Andreas Ketterer}\email{andreas.ketterer@uni-siegen.de}
\thanks{\\A.K. and A.L.-F. contributed equally.}

\affiliation{Naturwissenschaftlich-Technische Fakult\"at, Universit\"at Siegen, Walter-Flex-Str. 3, 57068 Siegen, Germany}
\affiliation{Laboratoire Mat\'eriaux et Ph\'enom\`enes Quantiques, Sorbonne Paris Cit\'e, Universit\'e Paris Diderot, CNRS UMR 7162, 75013 Paris, France}

\author{Adrien Laversanne-Finot}\email{adrien.laversanne-finot@univ-paris-diderot.fr}
\affiliation{Laboratoire Mat\'eriaux et Ph\'enom\`enes Quantiques, Sorbonne Paris Cit\'e, Universit\'e Paris Diderot, CNRS UMR 7162, 75013 Paris, France}

\author{Leandro Aolita}\email{aolita@if.ufrj.br}
\affiliation{Instituto de F\'{\i}sica, Universidade Federal do Rio de Janeiro, Caixa Postal 68528, 21941-972 Rio de Janeiro, RJ, Brazil}
 \affiliation{{ICTP South American Institute for Fundamental Reserch
Instituto de F\'isica Te\'orica, UNESP-Universidade Estadual Paulista R. Dr. Bento T. Ferraz 271, Bl. II, S\~ao Paulo 01140-070, SP, Brazil}}

\date{\today} 

\begin{abstract}
Supraquantum nonlocality refers to correlations that are more nonlocal than allowed by quantum theory but still physically conceivable in post-quantum theories, in the sense of respecting the basic no-faster-than-light communication principle. While supraquantum correlations are relatively well understood for finite-dimensional systems, little is known in the infinite-dimensional case.
Here, we study supraquantum nonlocality for bipartite systems with two measurement settings and infinitely many outcomes per subsystem. We develop a formalism for generic no-signaling black-box measurement devices with continuous outputs in terms of probability measures, instead of  probability distributions, which involves a few technical subtleties. We show the existence of a class of supraquantum Gaussian correlations, which violate the Tsirelson bound of an adequate continuous-variable Bell inequality. We then introduce the continuous-variable version of the celebrated Popescu-Rohrlich (PR) boxes, as a limiting case of the above-mentioned Gaussian ones. Finally, we perform a  characterisation of the geometry of the set of continuous-variable no-signaling correlations. Namely, we show that that the convex hull of the continuous-variable PR boxes is dense in the no-signaling set. {We also show that these boxes are extreme in the set of no-signaling behaviours and {provide} evidence {suggesting} that they are indeed the only extreme points of the no-signaling set.} Our results lay the grounds for studying generalized-probability theories in continuous-variable systems.
\end{abstract}

\pacs{03.65.Ud, 03.65.Ta}

\maketitle 
%

%

\section{Introduction}
\emph{Bell nonlocality} refers to correlations incompatible with local hidden-variable theories \cite{BELL}, which explain correlations between space-like separated measurement outcomes as due exclusively to past common causes. 
Since the pioneering works of Bell  \cite{BELL}, and of Clauser, Horn, Shimony and Holt \cite{CHSH}, it is known that quantum mechanics admits Bell nonlocality, i.e., that local measurements on quantum entangled states produce Bell nonlocal correlations.
However, nonlocality is not a phenomenon exclusive  of quantum theory. Hypothetical supraquantum theories satisfying the basic \emph{no-signaling principle} of no-faster-than-light communication, in consistency with special relativity, can produce Bell correlations that are even more nonlocal than those compatible with quantum theory. This is generally referred to as \emph{supraquantum Bell nonlocality}.
The first known example thereof was the so-called Popescu-Rohrlich (PR) boxes \cite{PRboxes}. These are hypothetic black-box measurement devices that can violate the Clauser-Horn-Shimony-Holt inequality up to its algebraic maximum of 4, which is above the maximum value of $2\sqrt{2}$ attained by quantum correlations, known as Tsirelson's bound \cite{Cirelson}.

Importantly, the aim of studying supraquantum nonlocality is by no means  to question the validity of quantum mechanics, but, rather on the contrary,  actually to gain a better understanding of quantum nonlocality itself. For instance, even though unphysical, PR boxes make excellent \emph{units of Bell nonlocality}, serving, in fact, as reference to quantify the nonlocal weight of quantum correlations \cite{EPR2, BKP06}. 
Furthermore, understanding why quantum mechanics is not as nonlocal as allowed by the no-signalling principle gives us valuable insights with foundational implications on the very axiomatic structure of quantum theory. For instance, a seminal result in this direction was the realization that  the physical existence of PR boxes would make communication complexity problems trivial \cite{vanDam2000,vanDam2005,Brassard06}, which is a highly implausible possibility. Hence, if one accepts that communication complexity is not trivial as a postulate, the non-existence of PR boxes is implied. In fact, in a similar spirit, a large effort has been devoted to proposing physically reasonable postulates from which Tsirelson's bound can be derived from first principles (see, e.g., Refs. \cite{Pawlowski09, Navascues09, OW10, Gallego11,Fritz13}). 

PR boxes have been generalized to arbitrary finite numbers of measurement outcomes \cite{PRABoxesDdimesnions} and to multipartite systems as well \cite{PRmultipartite}.
What is more, in the multipartite scenario, non-trivial tight Bell inequalities are known without a quantum violation, i.e. for which the  quantum maximum coincides with the local one and is below the no-signalling one \cite{Almeida10}.
In addition, supraquantum nonlocality has been explored even in the bipartite scenario where only one part makes measurements \cite{Sainz15}.
From a broader perspective, Bell nonlocality in generalised probabilistic theories has been extensively studied in the finite-dimensional case (see \cite{ReviewNonlocality} and Refs. therein).
Nevertheless, in striking contrast, essentially nothing is known about supraquantum nonlocality in continuous-variable (CV) systems. On the one hand, this is surprising in view of the huge amount of work on CV quantum nonlocality (see, e.g., Refs. \cite{Grangier, Tan, Gilchrist, Munro,LeeJaksch,CFRD, CFRDmultiSet, SignBinning, RootBinning, CFRDSalles, CFRDSallesLong}) and the importance of CV systems for quantum information processing \cite{Braunstein05, Ferraro05,Weedbrook12}. On the  other hand, this is at the same time understandable because, for CV systems, the set of local correlations (as well as that of no-signalling ones) is a generic convex set, instead of  a (computationally much tamer) convex polytope as in finite-dimensional systems \cite{Zukowski99,LinProgramBellineq, Masanes}.

In this article we conduct an exploration of CV supraquantum nonlocality. To begin with, we develop a formalism to deal with generic no-signalling black-box measurement devices with discrete measurement settings (inputs) and CV measurement outcomes (outputs). The correlations produced by such devices are described by probability measures instead of probability distributions. We then show the existence of a class of supraquantum Gaussian PR boxes, for bipartite systems with dichotomic inputs and real, continuous outputs. This {is} done by showing that these behaviours violate the Calvalcanti-Foster-Reid-Drummond (CFRD) inequality \cite{CFRD}, which admits no quantum violation in the bipartite case \cite{CFRDSallesLong}. Next, we introduce, a limiting case of the supraquantum Gaussian behaviours, a hierarchy of CV PR boxes, whose ground level consists of local, deterministic points and the upper levels of nonlocal, non-deterministic ones. The CV PR boxes obtained are very similar in structure to the finite-dimensional ones. {To end up with, we perform a characterization of the set of CV no-signaling behaviours and show that all CV PR boxes are extreme points of the CV no-signaling set, and that their convex hull (i.e. the set of all finite convex sums) is dense therein. In particular, we discuss whether the CV PR boxes are the only extreme no-signaling behaviours and, along with some evidence, conjecture that this is indeed the case.} 

The paper is structured as follows. In Sec. \ref{sec:preliminaries}, we set up the mathematical framework for CV no-signaling behaviours based on probability measures. In Sec.~\ref{sec:CVnonlocalBoxes}, we introduce the supraquantum Bell nonlocal Gaussian behaviours and the CV PR boxes. Sec.~\ref{sec:Characterization} is devoted to the geometrical characterization of the set of CV no-signaling set. Finally, we conclude, in Sec.~\ref{sec:Conclusion}, with some final remarks and perspectives of our work.

\section{Preliminaries: mathematical representation of CV Bell correlations}
\label{sec:preliminaries}

We consider a bipartite Bell experiment where two space-like separated observers, conventionally referred to as Alice (A) and Bob (B), make measurements. We work in the generic device-independent scenario where the measurement apparatuses are treated as unknown black-box measurement devices [see Fig.~\ref{fig:AliceBobSets}(a)]. Alice's (Bob's) device has a dichotomic input  $x$ ($y$) $\in\{0,1\}$ and a continuous output $a$ ($b$) $\in\mathbb R$. %
That is, we are considering infinite resolution:  we want to investigate the ideal situation where the outputs can take any arbitrary real value.
The statistics produced by such devices is most conveniently described in terms of probability measures, which we briefly recap in what follows.
We consider probability spaces defined by a triple $\{\Omega,\mathcal B(\Omega),\mu\}$, where $\Omega$ denotes a sample space, $\mathcal B(\Omega)$ the Borel $\sigma$-algebra of events on $\Omega$ 
(i.e., the smallest $\sigma$-algebra that contains all open subsets of $\Omega$) 
and $\mu:\mathcal B(\Omega)\rightarrow [0,1]$ a Borel probability measure. In our case, the sample space is given by a product space $\Omega=\Omega_\mathrm{A}\times\Omega_\mathrm{B}$, with $\Omega_\mathrm{A}=\Omega_\mathrm{B}=\mathbb R$, where the first and second factors, $\Omega_\mathrm{A}$ and $\Omega_\mathrm{B}$, correspond to the outputs of A and B, respectively.
The probability measure $\mu$ is required to  
be normalized, $\mu(\mathbb R\times \mathbb R)=1$, 
and to satisfy the additivity property $\mu\left(\bigcup_{i=1} E_i\right)=\sum_{i=1} \mu(E_i)$, for every countable sequence $\{E_i\}_i$ of disjoint events $E_i\in\mathcal B(\mathbb R\times \mathbb R)$, where $\cup$ stands for the set union. The probability of an event $E\in \mathcal B(\mathbb R\times \mathbb R)$ is then given by $P(E):=\mu(E)$. 
We denote the set of all probability measures on $\mathcal B(\mathbb R\times \mathbb R)$ as $\mathcal{M}_{\mathbb R\times\mathbb R}$. 

\begin{figure}[t]
\begin{center}
\includegraphics[width=0.47\textwidth]{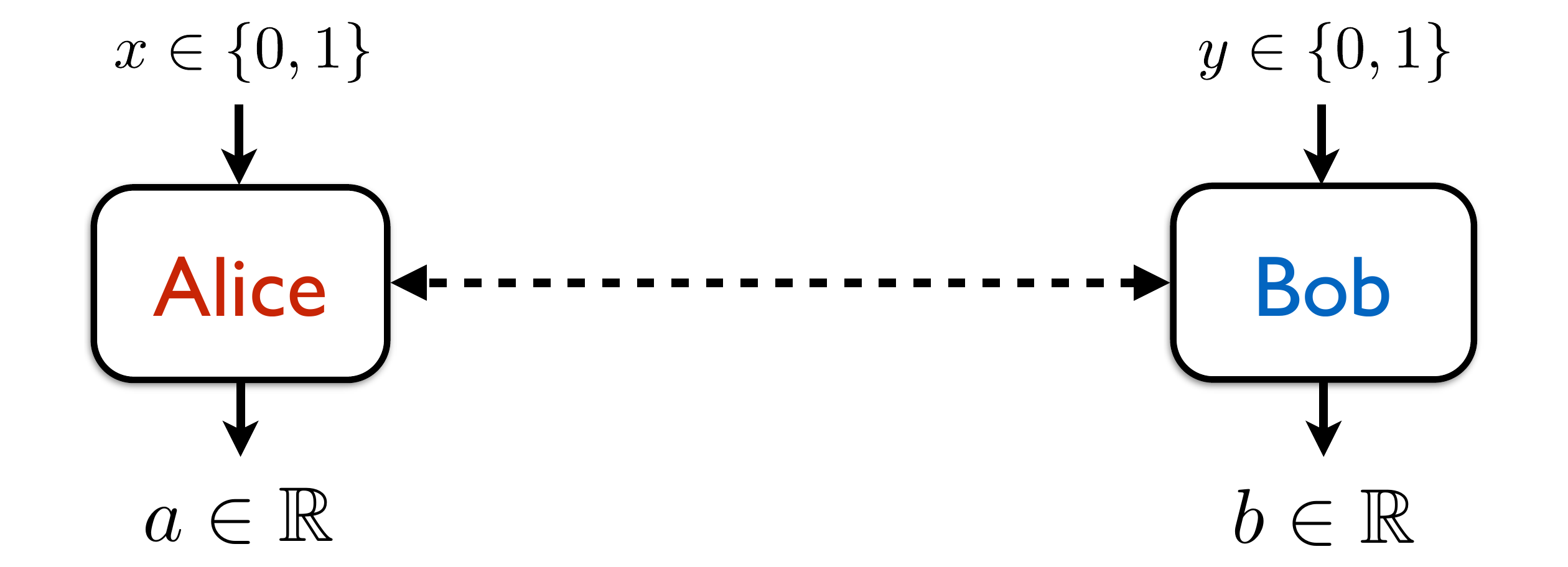}
\end{center}
\caption{(Color online) Schematic representation of a bipartite Bell experiment with continuous measurement outcomes in the so-called device-independent scenario of black-box measurement instruments. 
Two space-like separated observers, Alice (A) and (Bob), perform local measurements on their subsystems with dichotomic measurements choices (inputs) $x$ and $y$, respectively, and obtain continuous-variable measurement outcomes (outputs) $a$ and $b$.}
\label{fig:AliceBobSets}
\end{figure}

The connection between a probability measure $\mu$ and a probability density $p$ (with respect to the Lebesgue measure) can be made explicit in the integral representation
\begin{align}
\mu(A\times B):=\int_{A\times B} d\mu(a',b')=\int_{A} \int_{B} p(a',b')\,da'\,db',
\label{eq:DefMeasFromDensity}
\end{align}
where $A\times B\in \mathcal B(\mathbb R\times \mathbb R)$, with $A, B\in \mathcal B(\mathbb R)$, $p(a',b')$ denotes the corresponding probability density to $\mu$, and $d\mu(a,b)$ and $da'\,db'$ refer to integrations with respect to $\mu$ and the Lebesgue measure on $\mathbb R\times \mathbb R$, respectively. 
Note that not every probability measure can be expressed in terms of a probability density as in Eq.~(\ref{eq:DefMeasFromDensity}). The question of the existence of a probability density is answered by the Radon-Nikodym (RN) theorem{,} whose statement {is} briefly review{ed} in Appendix~\ref{app:Radon}. While most assumptions of the RN theorem are fulfilled by any probability measure on $\mathbb R\times \mathbb R$, for us the crucial prerequisite is that $\mu$ has to be absolutely continuous with respect to the Lebesgue measure. However, as we will see later on, absolute continuity cannot be guaranteed for all types of probability measures which will become important when dealing with so-called boxes describing {idealized} unphysical outcome scenarios. 
Hence, all in all{,} it is {both} more general {and} more convenient to work with measures, as one {needs} no{t w}orry about the existence of a density.  

We thus arrive at the following definition.
\begin{definition}[CV Bell behaviour] 
A behaviour is a joint conditional probability measure represented by a $2\times2$ matrix $\boldsymbol\mu=\{{\mu_{x,y}}\}_{x,y\in\{0,1\}}$ with arbitrary probability measures $[\boldsymbol\mu]_{x,y}:={\mu_{x,y}}\in \mathcal{M}_{\mathbb R\times\mathbb R}$ as entries. The set of all behaviours is denoted as $\mathcal{M}_{\mathbb R\times\mathbb R}^4$. 
\end{definition}

Note that, for finite-dimensional systems, the sample space has a finite number of events, so that joint conditional probability measures reduce to the more usual notion of joint conditional probability distributions \cite{ReviewNonlocality}. 
Also as in the discrete case, since the observers are space-like separated, $\boldsymbol\mu$ must fulfill the no-signaling principle, given, in this language, by the constraints:
\begin{subequations}
\label{eq:NS_constraints}
\begin{align}
{\mu_{x,y}}(A\times \mathbb R) &= \mu_{x,\overline{y}}(A \times \mathbb R) \quad \forall \ x\in\{0,1\},   \label{eq:constraint1} \\
{\mu_{x,y}}(\mathbb R\times B) &= \mu_{\overline{x},y}(\mathbb R\times B) \quad \forall \ y\in\{0,1\},  \label{eq:constraint2}
\end{align}
\end{subequations}
for all $A, B\in \mathcal B(\mathbb R)$, where $\overline{y}=y\oplus1$ and  $\overline{x}=x\oplus1$, with $\oplus$ the sum modulo 2. 

Conditions~(\ref{eq:constraint1}) and (\ref{eq:constraint2}) imply respectively that Alice's and Bob's marginal measures $\mu_{x}(A):={\mu_{x,y}}(A\times \mathbb R)$ and $\mu_{y}(B):={\mu_{x,y}}(\mathbb R\times B)$ are independent of each others' input, which prevents signaling. 
We call any  $\boldsymbol\mu$ satisfying these conditions a \emph{no-signaling behaviour}, and denote the set of all no-signaling behaviours by  $\mathcal{M}_\mathrm{NS}\subset\mathcal{M}_{\mathbb R\times\mathbb R}^4$. 

Quantum correlations, in turn, are those described by the behaviours that can be
expressed as 
\begin{align}
 \label{eq:quant_Behav}
 {\mu_{x,y}}(A\times B) &= \Tr\left[M_x{(A)}\otimes M_y{(B)}\,\varrho_{AB} \right], 
\end{align}
for all $A, B\in \mathcal B(\mathbb R)$, where $\rho_{AB}$ is an arbitrary bipartite quantum state on a Hilbert space $\mathcal H:=\mathcal H_\mathrm{A}\otimes\mathcal H_\mathrm{B}$, with $\mathcal H_\mathrm{A}$ and $\mathcal H_\mathrm{B}$ the local Hilbert spaces of Alice's and Bob's systems, respectively, and $M_x$ and $M_y$ are,  for all $x\ (y)\in\{0,1\}$,  semi-spectral measures, also known as positive-operator valued measures (POVMs) \cite{QuantumMeasurement}.  
The latter means that $M_x,M_y:\mathcal B(\mathbb R)\rightarrow \mathcal L_{\geq 0}(\mathcal H)$ are maps such that, for all $A\ (B)\in \mathcal B(\mathbb R)$, $M_x{(A)}\in\mathcal L_{\geq 0}(\mathcal H_\mathrm{A})$ [$M_y{(B)}\in\mathcal L_{\geq 0}(\mathcal H_\mathrm{B})$], with $\mathcal L_{\geq 0}(\mathcal H_\mathrm{A})$ [$\mathcal L_{\geq 0}(\mathcal H_\mathrm{B})$] the space of positive semi-definite operators on $\mathcal H_\mathrm{A}$ ($\mathcal H_\mathrm{B}$); and that $M_x(\mathbb R)=\mathbb 1_\mathrm{A}$ ($M_y(\mathbb R)=\mathbb 1_\mathrm{B}$), with $\openone_\mathrm{A}$ ($\openone_\mathrm{B}$) the identity operator on $\mathcal H_\mathrm{A}$ ($\mathcal H_\mathrm{B}$). 
We call any  $\boldsymbol\mu$ satisfying Eq. \eqref{eq:quant_Behav} a \emph{quantum behaviour}, and denote the set of all quantum behaviours by $\mathcal{M}_\mathrm{Q}$. For generic Bell scenarios, the relationship $\mathcal{M}_\mathrm{Q}\subseteq\mathcal{M}_\mathrm{NS}$ holds. For the scenario under consideration here, we show below that $\mathcal{M}_\mathrm{Q}\subset\mathcal{M}_\mathrm{NS}$. We call any $\boldsymbol{\mu}\in\mathcal{M}_\mathrm{NS}\setminus\mathcal{M}_\mathrm{Q}$ a \emph{supraquantum behaviour}.

\begin{figure}[t]
\begin{center}
\includegraphics[width=0.47\textwidth]{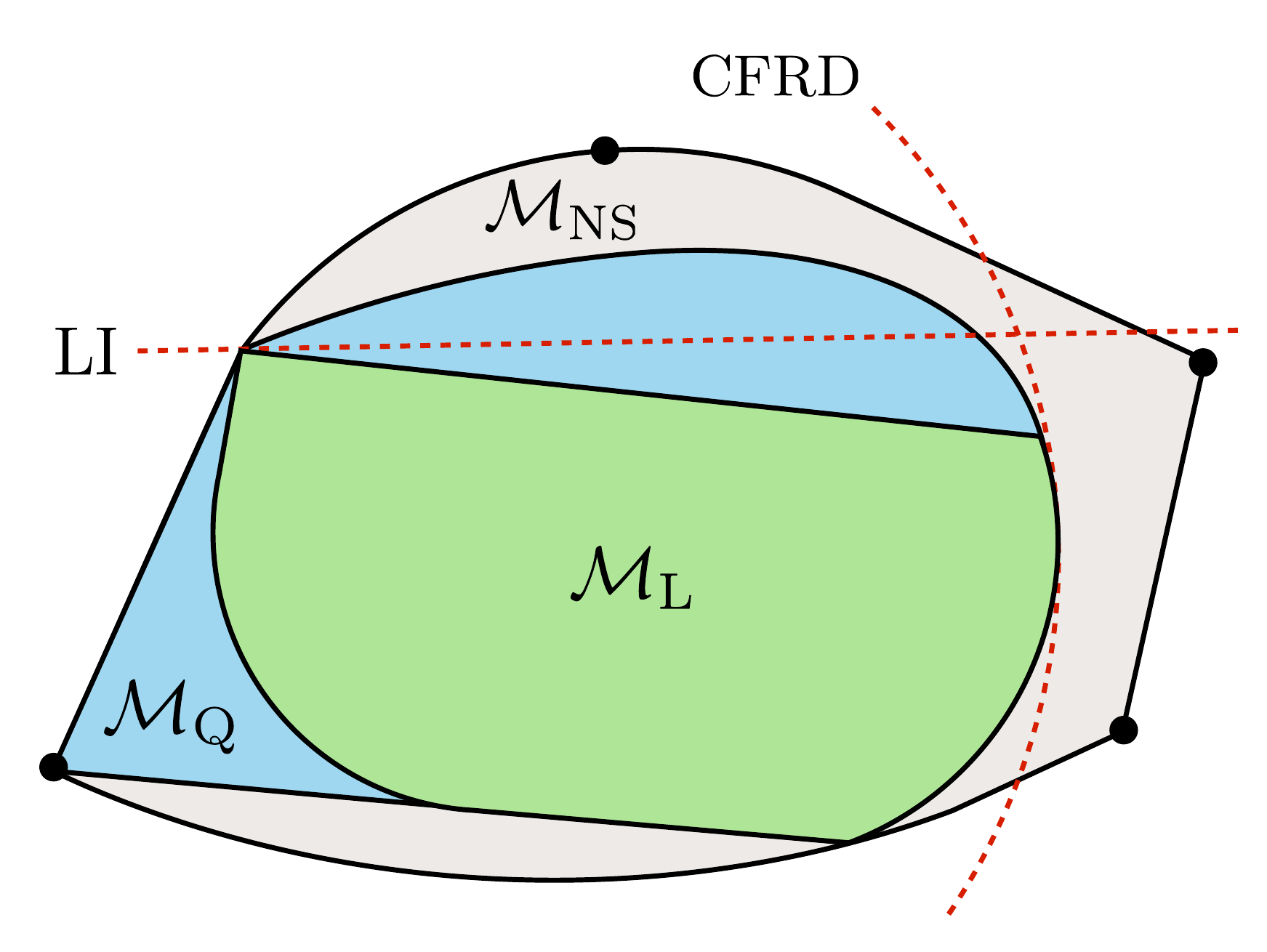}%
\end{center}
\caption{(Color online) Pictorial {(not rigorous)} geometrical representation of the (possible) inner structure of the set $\mathcal{M}_\mathrm{NS}$ of CV no-signaling behaviours in the Bell scenario of Fig. \ref{fig:AliceBobSets}.  $\mathcal{M}_\mathrm{NS}$ contains the set $\mathcal{M}_\mathcal Q$ of quantum behaviours, which contains, in turn, the set $\mathcal{M}_\mathcal L$ of local behaviours. 
All three sets are generic convex sets with infinitely many extreme points, delimited by facets as well as curved hyper-surfaces.
This is in contrast with the finite-dimensional case, where both  $\mathcal{M}_\mathrm{NS}$ and $\mathcal{M}_\mathcal L$ are convex polytopes, delimited exclusively by facets that can by characterisied by a finite number of linear Bell inequalities. In the plot, an example of a linear Bell inequality is represented as a straight line LI. 
Such linear inequality can, e.g., correspond to a Bell inequality for finite-dimensional systems, which can be violated by CV quantum correlations using so-called binning procedures \cite{Grangier, Tan, Gilchrist, Munro, SignBinning, RootBinning, LeeJaksch} (see also Refs. in \cite{ReviewNonlocality}).
Besides this, a hypothetical quantum extreme point is shown in the figure (light-blue corner). While such points are in principle possible, no explicit example thereof is known.
In this paper we consider a non-linear Bell inequality, the CFRD inequality \cite{CFRD}, represented as a curve in the plot. This inequality applies in the genuinely CV scenario of our interest and has, additionally, the appealing feature of admitting violations only by supraquantum behaviours (see Sec. \ref{sec:CVnonlocalBoxes}). Finally, four exemplary CV PR boxes are represented as extreme points of $\mathcal{M}_\mathrm{NS}$ (black dots).
}
\label{fig:NS_correlations}
\end{figure}

The last important class for our purposes is the one of classical correlations, described by the behaviours produced by local hidden-variable models:
\begin{align}
&\mu_{x,y}=\int_\Lambda \delta_{a(x,\lambda),b(y,\lambda)}\, d\eta(\lambda),
\label{eq:localbehaviours}
\end{align}
where $\lambda$ is the hidden variable, taking values in a parameter space $\Lambda$ according to a probability measure $\eta:\mathcal B(\Lambda)\rightarrow \mathbb R_{\geq0}$, and $\delta_{a(x,\lambda),b(y,\lambda)}$ is the CV version of the $\lambda$-th local deterministic response function. More precisely,  $\delta_{a,b}$ denotes the Dirac measure at the point $(a,b)\in\mathbb R^2$, i.e. the deterministic measure such that
\begin{align}
\delta_{a,b}(A\times B):=\left\{
                \begin{array}{lll}
                  1 & &\text{if } a \in A\text{ and } b\in B, \\
                  0 & &\mathrm{otherwise},
                \end{array}
                \right.
\label{eq:DiracMeasure}
\end{align}
for all $A, B\in \mathcal B(\mathbb R)$. In turn, for each $\lambda\in\Lambda$, $a(x,\lambda)$ and $b(y,\lambda)$ are respectively deterministic functions of $x$ and $y$, in a similar spirit to the local deterministic response functions in finite-dimensional scenarios \cite{ReviewNonlocality}. Since the outputs are locally generated from each input and the pre-established classical correlations encoded in $\lambda$, one typically calls any $\boldsymbol\mu$ given by Eq. \eqref{eq:localbehaviours} a \emph{local behaviour}. We denote the set of all local behaviours by $\mathcal{M}_\mathrm{L}\subseteq\mathcal{M}_\mathrm{Q}$. In turn, any $\boldsymbol{\mu}\in\mathcal{M}_\mathrm{NS}\setminus\mathcal{M}_\mathrm{L}$ is a \emph{nonlocal behaviour}.

Finally, we emphasise that, in contrast to the finite-dimensional case, $\mathcal{M}_\mathrm{L}$ does not define a polytope (i.e., a convex set with finitely many extreme points), see Fig. \ref{fig:NS_correlations}. This is due to the fact that Dirac measures are extreme in $\mathcal{M}_\mathrm{NS}$ and $\mathcal{M}_\mathrm{L}$ is generated by a continuously infinite number of them. It follows, then, that $\mathcal{M}_\mathcal L$ cannot be characterized by a finite set of linear Bell inequalities \cite{Zukowski99,LinProgramBellineq,Masanes}. In the next section, we use a non-linear Bell inequality to identify not only nonlocal behaviours but supraquantum ones.

\section{Continous-variable supraquantum nonlocality}
\label{sec:CVnonlocalBoxes}
In Ref.~\cite{CFRD}, Calvalcanti, Foster, Reid, and Drummond derived the nonlinear Bell inequality 
\begin{align}
&\left[\expec{A_0\,B_0}-\expec{A_1\,B_1}\right]^2  +\left[\expec{A_0\,B_1}+\expec{A_1\,B_0}\right]^2 \nonumber \\ 
&\leq \expec{A_0^{2}\,B_0^{2}}+\expec{A_0^{2}\,B_1^{2}}+\expec{A_1^{2}\,B_0^{2}}+\expec{A_1^{2}\,B_1^{2}},
\label{eq:CFRD}
\end{align}
where $A_0$ and $A_1$ ($B_0$ and $B_1$) are the real, continuous outputs of Alice's (Bob's) box for the inputs 0 and 1, respectively. Using the integral representation of Eq. \eqref{eq:DefMeasFromDensity}, the expectation values of such observables appearing in the inequality can be recast as cross-moments of the behaviour elements ${\mu_{x,y}}$:
\begin{align}
\expec{A_x^{n_a}\,B_y^{n_b}}=\int_{\mathbb R^2} a^{n_a}\, b^{n_b}\,  d{\mu_{x,y}}(a,b).
\label{eq:CrossMoments}
\end{align}
Eq.~(\ref{eq:CFRD}) can be generalised to higher number of parties~\cite{CFRD} as well as observables per party~\cite{CFRDmultiSet}. We refer to the bipartite dichotomic-input version of inequality, given by Eqs. \eqref{eq:CFRD} and \eqref{eq:CrossMoments}, as the \emph{CFRD inequality}. The  inequality has a number of interesting properties~\cite{CFRD,CFRDmultiSet}. Specially relevant for our purposes is the fact that it cannot be violated by any quantum bebaviour. This was first shown in Ref. \cite{CFRDSalles} for the restricted case of measurements of (quantum) phase-space quadrature operators, and then extended to the general case of arbitrary quantum measurements in Ref.  \cite{CFRDSallesLong}. Hence, the CFRD constitutes a non-trivial Bell inequality with no quantum violation. Any no-signalling behaviour that violates it is thus automatically certified as supraquantum, as we do next.

\begin{figure}[t]
\begin{center}
\includegraphics[width=0.485\textwidth]{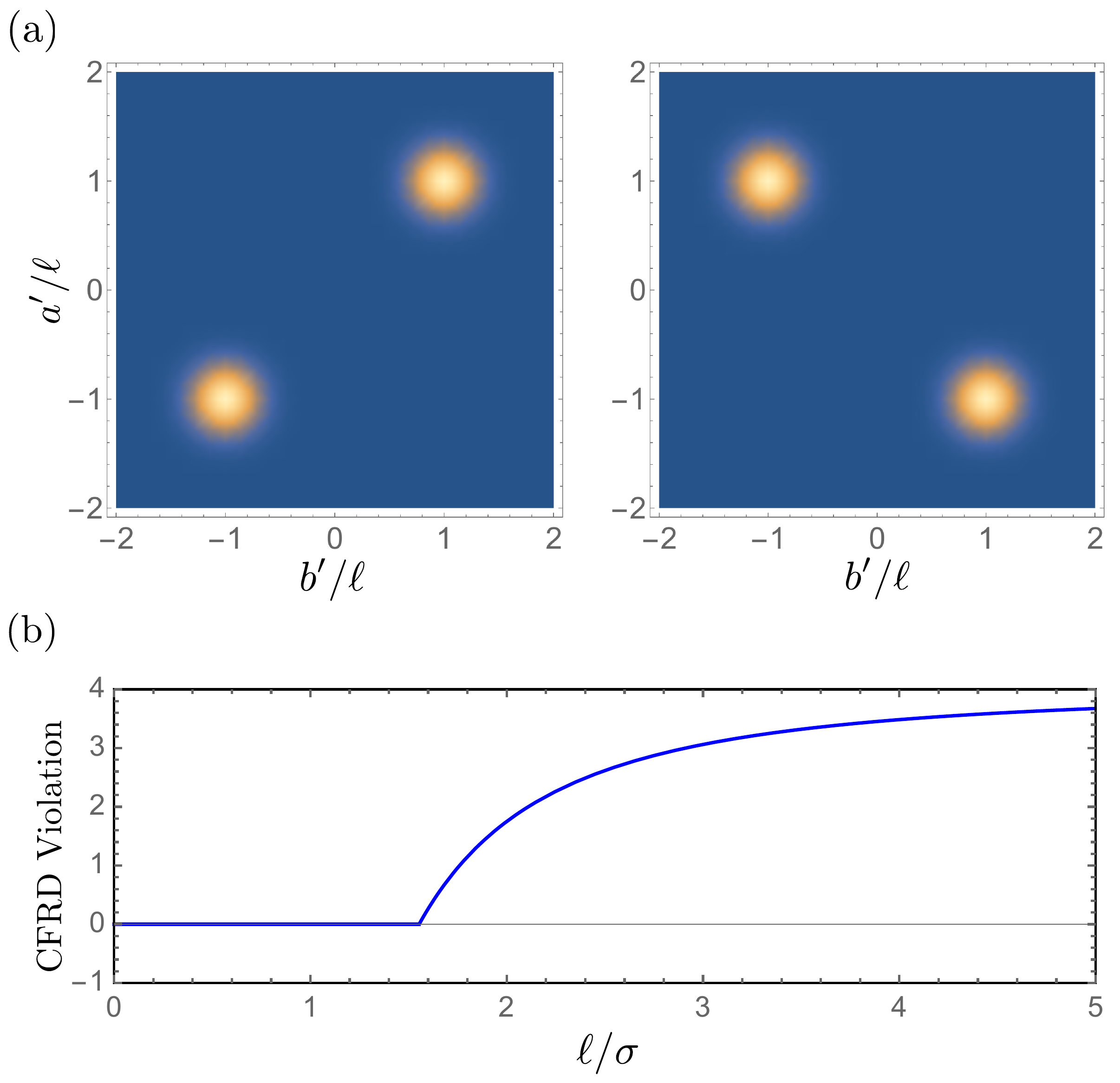}
\end{center}
\caption{(Color online) (a) Density plots of a Gaussian PR box of order 2, with centre vector characterised by $\boldsymbol{a}=(\ell,-\ell)=\boldsymbol{b}$ and width vector $\boldsymbol{\sigma}=(\ell/5,\ell/5)$, for the inputs $(x,y)=(0,0)$, $(0,1)$, or $(1,0)$ (left) and $(x,y)=(1,1)$ (right). Note that, for both plots, the projections onto the horizontal as well as vertical axes coincide, reflecting the fact that the behaviour is no-signalling. Each centre point may also have a different width (or squeezing), but we do not consider that here for simplicity. (b) Violation of the CFRD inequality, normalised by the factor $\ell^4$, by the Gaussian behaviour in question as a function of the parameter $\ell/\sigma$. The CFRD inequality certifies that the Gaussian PR box is supraquantum for the parameter region with $\ell/\sigma \geq\sqrt{1+\sqrt 2}\approx 1.55$. 
}
\label{fig:ViolCVPRd2}
\end{figure}

The first case that we study is a sub-class of behaviours that we term \emph{Gaussian PR boxes}.
To this end, we first introduce two real vectors, $\boldsymbol{a}:=(a_1,\ldots,a_k)$ and $\boldsymbol{b}:=(b_1,\ldots,b_k)$, with different components, i.e., such that $a_1\neq a_2\neq\ldots a_k$ and $b_1\neq b_2\neq\ldots b_k$, and one positive-real vector $\boldsymbol{\sigma}:=(\sigma_1,\ldots,\sigma_k)$, all of length $k\in\mathbb{N}$. The vectors $\boldsymbol{a}$ and $\boldsymbol{b}$ determine $k$ points $(a_j,b_j)$ where Gaussian-measure components are centred; while the vector $\boldsymbol{\sigma}:=(\sigma_1,\ldots,\sigma_k)$ determines their widths.
More precisely, then, we say that $\boldsymbol{\mu}\in\mathcal{M}_\mathrm{NS}$ is a Gaussian PR box of order $k$, with centre vector $(\boldsymbol{a},\boldsymbol{b})$ and width vector $\boldsymbol{\sigma}$, if it is of the form
\begin{align}
\mu_{x,y}^{(k,\boldsymbol{a}, \boldsymbol{b}, \boldsymbol{\sigma})}:= \frac{1}{k} \sum_{j=1}^{k} \mathcal{N}_{(a_j,b_{[j+x\,y]_k}),\sigma_j},
\end{align}
where $[\ ]_k$ denotes modulo $k$ and $\mathcal{N}_{(a, b),\sigma}$ is the normal (Gaussian) measure centred at $(a, b)$ and with width $\sigma$, defined through Eq.~(\ref{eq:DefMeasFromDensity}) with the probability density
\begin{align}
p_{(a,b),\sigma}(a',b')=\frac{1}{2\,\pi\, \sigma^2} e^{-\frac{(a-a')^2+(b-b')^2}{2\,\sigma^2}}.
\end{align}

Whether a Gaussian PR box is supraquantum or not depends on $(\boldsymbol{a},\boldsymbol{b})$ and $\boldsymbol{\sigma}$. As an example, consider next the simple case with $k=2$, $\boldsymbol{a}=(\ell,-\ell)=\boldsymbol{b}$, for some arbitrary $\ell\in\mathbb R_{\neq0}$, and $\boldsymbol{\sigma}=(\sigma,\sigma)$, graphically represented in Fig. \ref{fig:ViolCVPRd2} (a). It is immediate to see that the resulting behaviour violates the CFRD inequality by the amount
\begin{align}
\label{eq:Gaussian_violation}
\max\ \big\{8\,\ell^4-4\,\left(\sigma^2+\ell^2\right)^2,0\big\}.
\end{align}
This violation is plotted in Fig.~\ref{fig:ViolCVPRd2}(b) as a function of $\sigma/\ell$. Note that it grows unboundedly with $\ell$. The condition for this Gaussian PR box to violate the CFRD inequality is $\ell/\sigma\geq \sqrt{1+\sqrt 2}\approx 1.55$, as can be graphically appreciated in the figure. In turn, taking, for the Gaussian PR box above, the limit $\sigma\to0$,  one obtains the behaviour with components
\begin{align}
\mu_{x,y}=\frac{1}{2} \big[&\delta_{\ell,(-1)^{xy}\ell}+ \delta_{-\ell,-(-1)^{xy}\ell}\big],
\label{eq:BinaryCVPRbox}
\end{align}
with $\delta$ the measure defined in Eq. \eqref{eq:DiracMeasure}. This limiting box violates the CFRD inequality by $4\ell^4$. In fact, it is the CV version of the original dichotomic-input dichotomic output PR box \cite{PRboxes}.

Similarly, to define generic CV PR boxes, we take the $\boldsymbol{\sigma}\to\boldsymbol{0}$ limit of the Gaussian PR boxes of Eq. \eqref{eq:PRBoxesMeasures}. That is, we say that $\boldsymbol\mu^{(k,\boldsymbol{a}, \boldsymbol{b})}\in\mathcal{M}_\mathrm{NS}$ is a CV PR box of order $k$ and center vector $(\boldsymbol{a},\boldsymbol{b})$, with different real components such that $a_1\neq a_2\neq\ldots a_k$ and $b_1\neq b_2\neq\ldots b_k$, if it is of the form
\begin{align}
\mu_{x,y}^{(k,\boldsymbol{a}, \boldsymbol{b})}:=\mu_{x,y}^{(k,\boldsymbol{a}, \boldsymbol{b}, \boldsymbol{0})}= \frac{1}{k} \sum_{j=1}^{k} \delta_{a_j,b_{[j+x\,y]_k}}.
\label{eq:PRBoxesMeasures}
\end{align}
One can immediately verify that these behaviours fulfil the no-signalling constraints \eqref{eq:NS_constraints}. These boxes are the CV version of the finite-dimensional PR boxes generalized to arbitrarily many outputs and dichomotic inputs given in Ref. \cite{PRABoxesDdimesnions}. {Still, Eq. \eqref{eq:PRBoxesMeasures}  does not yet describe the most general CV PR box, because input and output relabelling symmetries must be taken into account. For dichotomic inputs, the possible local, reversible relabelings are given by $x\rightarrow [x+1]_2$ and $y\rightarrow [y+1]_2$  \cite{PRABoxesDdimesnions}. The situation is notably different, however, for the  outputs, as they are continuous. For CV outputs, the most general local, reversible relabelings are given by $\boldsymbol a\rightarrow \boldsymbol{\alpha}_x(\boldsymbol a)$ and $\boldsymbol b\rightarrow \boldsymbol{\beta}_y(\boldsymbol b)$, where $\boldsymbol{\alpha}_x:\mathbb R^k\rightarrow \mathbb R^k$ and $\boldsymbol{\beta}_y:\mathbb R^k\rightarrow \mathbb R^k$ are, for every $x,y\in\{0,1\}$, bijective maps from $\mathbb R^k$ to itself. This amounts to reshuffling the components of the center vectors in a reversible, input-dependent fashion, so that the condition $[\boldsymbol{\alpha}_x(\boldsymbol a)]_1\neq [\boldsymbol{\alpha}_x(\boldsymbol a)]_2\neq\ldots [\boldsymbol{\alpha}_x(\boldsymbol a)]_k$ and $[\boldsymbol{\beta}_y(\boldsymbol b)]_1\neq [\boldsymbol{\beta}_y(\boldsymbol b)]_2\neq\ldots [\boldsymbol{\beta}_y(\boldsymbol b)]_k$ is always maintained.

Since the relabelings are local and reversible, all boxes equivalent under them have the same nonlocality properties.
Indeed, all the boxes given by Eq. \eqref{eq:PRBoxesMeasures}, i.e. for all different center vectors, are equivalent under input-independent relabelings. So, any of them, i.e. for any fixed center vector, can be taken as representative to define (modulo local, reversible, and input-dependent relaballings) the entire class of all CV PR boxes. 
This is, in turn, equivalent to allowing for input-dependent center vectors $(\boldsymbol{a}_x,\boldsymbol{b}_y)$ directly in the definition:
\begin{definition}[Set of CV PR boxes]\label{def:SetPRboxes}
We define the class $\mathcal{M}_\mathrm{PR}$ as the set 
\begin{align}
\mathcal{M}_\mathrm{PR}:=\big\{\boldsymbol\mu^{(k,\boldsymbol{a}_{0},\boldsymbol{a}_{1},\boldsymbol{b}_{0},\boldsymbol{b}_{1})}\in\mathcal{M}_\mathrm{NS}\big\}_{k\in\mathbb N,\,\boldsymbol{a}_{0},\boldsymbol{a}_{1},\boldsymbol{b}_{0},\boldsymbol{b}_{1}\in \mathbb R^{k}},
\label{eq:SetCVPRmeasures}
\end{align}
where each behaviour component $\left[\boldsymbol\mu^{(k,\boldsymbol{a}_{0},\boldsymbol{a}_{1},\boldsymbol{b}_{0},\boldsymbol{b}_{1})}\right]_{x,y}$ is given by a measure $\mu_{x,y}^{(k,\boldsymbol{a}_{x},\boldsymbol{b}_{y})}$ as in Eq.~\eqref{eq:PRBoxesMeasures}, with a possibly different vector $(\boldsymbol{a}_{x},\boldsymbol{b}_{y})$ for each $(x,y)\in\{0,1\}^2$.
\end{definition}

Note that for $k=1$, CV PR boxes reduce to local, deterministic behaviours, whose components are given by Dirac delta measures. In contrast, for all $k\geq 2$, Def. \ref{def:SetPRboxes} yields non-local, non-deterministic behaviours. Here, for simplicity, we use the term ``CV PR box"  for all $k\in\mathbb N$ indistinctly, the distinction between  local, deterministic and non-local, non-deterministic ones being given by the order 
$k$. In the next section, we show that every element of $\mathcal{M}_\mathrm{PR}$ is an extreme behaviour of $\mathcal{M}_\mathrm{NS}$ and that the convex hull of  $\mathcal{M}_\mathrm{PR}$ is dense in $\mathcal{M}_\mathrm{NS}$.

\section{Characterization of the set of no-signalling behaviours}
\label{sec:Characterization}
We start by recapping basic definitions of convex combinations and extremality. The convex hull $\mathrm{Conv}(\mathcal{M})$ of an arbitrary (finite or infinite) set $\mathcal{M}$ of behaviours is the set of all finite convex sums of elements of $\mathcal{M}$:
\begin{align}
\mathrm{Conv}(\mathcal{M})=\Big\{\sum_{i=1}^n q_i\, \boldsymbol{\mu}_i :& \boldsymbol{\mu}_i\in\mathcal{M} \Big\}_{q_i\geq 0,\, \sum_{i=1}^n\, q_i=1,\ n\in\mathbb N}.
\label{eq:ConvexHull}
\end{align}
In turn, if $\mathcal{M}$ contains an uncountably infinite number of elements, continuous convex combinations (i.e., convex integrals) of infinitely many elements can be considered too but are not {necessarily} contained in $\mathrm{Conv}(\mathcal M)$.

Clearly, any behaviour that admits a decomposition in terms of a convex integral of uncountably infinitely many behaviours, admits also a decomposition in terms of a convex sum of finitely many behaviour. Similarly, any behaviour that admits a decomposition in terms of a convex sum of an arbitrary finite number of behaviours admits also a decomposition in terms of a convex sum of two behaviours.
This leads us to the same definition of extreme no-signaling behaviours as in discrete variables.
\begin{definition}[Extreme no-signaling behaviours] We call $\boldsymbol{\mu}$ an \emph{extreme point} of $\mathcal{M}_\mathrm{NS}$ if, for {any} $\boldsymbol{\mu}^*,\boldsymbol{\mu}'\in\mathcal{M}_\mathrm{NS}$ and $0\leq q\leq 1$ {such that} 
\begin{align}
\boldsymbol{\mu}=q\, \boldsymbol{\mu}^*+(1-q)\,\boldsymbol{\mu}'
\label{eq:NS_extremal}
\end{align}
{it holds that} either $q=1$ and $\boldsymbol{\mu}^*=\boldsymbol{\mu}$, or  $q=0$ and $\boldsymbol{\mu}'=\boldsymbol{\mu}$.
\end{definition}

Now, we know that every $\boldsymbol{\mu}\in\mathcal{M}_\mathrm{PR}$ has a finite number of outcomes with non-zero probability {and belongs to $\mathcal{M}_\mathrm{NS}$}. That is,  $\boldsymbol{\mu}$ is either an extreme point of $\mathcal{M}_\mathrm{NS}$ or it can be decomposed as the convex sum of at most finitely many points in $\mathcal{M}_\mathrm{NS}$. However, the fact that finite-dimensional PR boxes are no-signalling extreme implies that the former is the case. This follows from the fact that finite-dimensional PR boxes are given by an equivalent expression to that in Eq. \eqref{eq:PRBoxesMeasures} where Kronecker deltas are in the place of the Dirac ones \cite{PRABoxesDdimesnions}. This proves, then, that all CV PR boxes are no-signaling extreme:
\begin{observation}[Extremality of $\mathcal{M}_\mathrm{PR}$] All elements of $\mathcal{M}_\mathrm{PR}$ are extreme points of $\mathcal{M}_\mathrm{NS}$.
\label{observation}
\end{observation}

Observation \ref{observation} constitutes, in turn, a generalisation to the CV realm of the result of Ref. \cite{Pironio2005}, where it is shown that any extreme point of the no-signaling set with a given finite number of inputs and outputs is also extreme in the no-signaling set with any higher (but still finite) number of inputs and outputs. 
In addition, since $\mathcal{M}_\mathrm{PR}$ is not finite, the observation also directly implies that $\mathcal{M}_\mathrm{NS}$ is not a polytope. On the other hand, the fact that $\mathcal{M}_\mathrm{NS}$ contains behaviours with infinitely many outcomes with non-zero probability (e.g., the Gaussian PR boxes of the previous section) automatically implies that $\mathcal{M}_\mathrm{NS}\not\subseteq\mathrm{Conv}(\mathcal{M}_\mathrm{PR})$, in {striking} contrast with the finite-dimensional case. This is due to the fact that every behaviour in $\mathrm{Conv}(\mathcal{M}_\mathrm{PR})$  necessarily has only finitely many outcomes with non-zero probability. 
Nevertheless, we show in App.~\ref{app:proofT1} that $\mathcal{M}_\mathrm{NS}$ is approximated arbitrarily well by $\mathrm{Conv}(\mathcal{M}_\mathrm{PR})$, in the formal sense of there existing, for all $\bm{\mu}\in \mathcal{M}_\mathrm{NS}$, a sequence of elements in $\mathrm{Conv}(\mathcal{M}_\mathrm{PR})$ that converges to $\bm{\mu}$. This proves the following.
\begin{theorem}[$\mathrm{Conv}(\mathcal{M}_\mathrm{PR})$ dense in $\mathcal{M}_\mathrm{NS}$] The closure $\overline{\mathrm{Conv}(\mathcal{M}_\mathrm{PR})}$ of $\mathrm{Conv}(\mathcal{M}_\mathrm{PR})$ equals $\mathcal{M}_\mathrm{NS}$. In other words, $\mathrm{Conv}(\mathcal{M}_\mathrm{PR})$ is a \emph{dense subset} of $\mathcal{M}_\mathrm{NS}$.
\label{theorem1}
\end{theorem}
{The theorem is proven in detail in App.~\ref{app:proofT1}. Let us sketch the proof idea  here. We} consider first the case of behaviours defined on a compact domain $[-\mathcal K, \mathcal K]^2$. 
{There}, we can use standard techniques from measure theory to show that for any no-signaling behaviour $\boldsymbol\mu$ one can find a sequence of convex sums of CV PR boxes $\boldsymbol\mu_n $ that converges to it. The main idea is {then} to define the considered sequence in such a way that its components become good approximations of the components of $\boldsymbol\mu$, in the limit of large $n$. This procedure can be seen as a generalization to the approximation of a function by piece-wise constant functions as it is used in integration theory{.  
Next, one generalizes this further} to an infinite sequence of compact intervals which, in the {infinite-length} limit, covers the whole space $\mathbb R\times \mathbb R${.} 

Even though $\mathcal{M}_\mathrm{PR}$ consists exclusively of extreme points of $\mathcal{M}_\mathrm{NS}$, the fact that $\mathrm{Conv}(\mathcal{M}_\mathrm{PR})$ is a {strict} subset of $\mathcal{M}_\mathrm{NS}$ in principle leaves room for other extreme points in $\mathcal{M}_\mathrm{NS}$ that are not contained in $\mathcal{M}_\mathrm{PR}$. {In the following, we approach this problem systematically by focusing first on behaviours with compact support. In this case, a related problem  
was addressed by D. Milman, who proved that, given a compact convex subset $\mathcal C$ of a locally convex space $\mathcal E$ (see \cite{simon2011convexity} for a definition of locally convex) and another set $\mathcal T\subset \mathcal C$ such that $\overline{\mathrm{Conv}(\mathcal{T})}=\mathcal C$, it follows that all extreme points of $\mathcal C$ are in the closure of $\mathcal T$ \cite{simon2011convexity}. The space $\mathcal{M}_{[-\mathcal K, \mathcal K]^2}$ of probability measures with bounded domain $[-\mathcal K, \mathcal K]^2 \subset \mathbb R^2$, is a compact subset of the locally convex space of all measures on the same domain.  The same holds also for the set of behaviours $\mathcal{M}_{[-\mathcal K,\mathcal K]^2}^4$. 
Moreover, the set of no-signaling behaviours on $[-\mathcal K, \mathcal K]^2$ is a closed subset of $\mathcal{M}_{[-\mathcal K, \mathcal K]^2}^4$ and thus also compact, which enables us to use Milman's theorem to characterize its extreme points.
In what follows, we deal with no-signaling and PR box behaviours on a compact domain. To emphasize this, we equip the corresponding no-signaling set and the set of CV PR boxes with a superscript $\mathcal K$, i.e. $\mathcal{M}^{(\mathcal K)}_\mathrm{NS}$ and $\mathcal{M}^{(\mathcal K)}_\mathrm{PR}$. 
Consequently, we arrive at the corollary:
\begin{corollary}[Characterization of $\mathcal{M}^{(\mathcal K)}_\mathrm{NS}$]
Every extreme point of $\mathcal{M}^{(\mathcal K)}_\mathrm{NS}$ belongs to the closure of ${\mathcal{M}^{(\mathcal K)}_\mathrm{PR}}$.
\end{corollary}

Further on, it is interesting to investigate if the closure of $\mathcal M^{(\mathcal K)}_{\mathrm{PR}}$ contains behaviours that are extreme as well. If this was not the case, it would prove that all extreme points of $\mathcal M^{(\mathcal K)}_{\mathrm{NS}}$ are in $\mathcal M^{(\mathcal K)}_{\mathrm{PR}}$. We thus have to answer the question if PR boxes of infinite order, i.e. in the limit $k \rightarrow \infty$ (see Eq.~(\ref{eq:PRBoxesMeasures})), are also extreme. In Appendix~\ref{app:limit_k} we provide evidence suggesting that this is not the case. More precisely, we provide an examplary sequence of PR boxes whose limiting behaviour is not extreme, thus implying that ${\mathcal{M}^{(\mathcal K)}_\mathrm{PR}}$ is not a closed set. 
This evidence leads us to the following conjecture.
\begin{conjecture}[Characterization of $\mathcal{M}^{(\mathcal K)}_\mathrm{NS}$] Every extreme point of $\mathcal{M}^{(\mathcal K)}_\mathrm{NS}$ belongs to 
$\mathcal{M}^{(\mathcal K)}_\mathrm{PR}$.
\label{conjecture}
\end{conjecture}

Even though, the preceding discussion was restricted to behaviours with outcomes on a compact set, we have reason{s} to believe that the conjecture holds also in the general case of unbounded support. Namely, in probability theory it is a rather standard result that all extreme points of the set of probability measures are given by Dirac measures (see Eq.~(\ref{eq:DiracMeasure})). In particular, this is the case for probability measures defined on $\mathbb{R}$. Similarly, the extreme no-signaling behaviours may have also only finite support, which would suggest our Conjecture~\ref{conjecture} also in the general case of behaviours defined on $\mathbb{R}$. A proof of Conjecture~\ref{conjecture} would however require more involved arguments which go beyond the scope of the present article.}

Let us finish with some final clarifications on the boundary and the boundedness of $\mathcal{M}_\mathrm{NS}$. In the finite-dimensional case{, the} boundar{y between the no-signaling behaviours and behaviour-like objects that still satisfy the no-signaling constraints but involve non-positive probability distributions} is given by the subset of all convex combinations of no-signaling extreme points resulting in non strictly-positive behaviours (i.e., whose $(x,y)$-th components are probability measures assigning zero probability to some event). Consequently, the set of no-signaling behaviours has a nonempty interior. In contrast, for infinite dimensional behaviours, the boundary of $\mathcal{M}_\mathrm{NS}$ is actually $\mathcal{M}_\mathrm{NS}$ itself showing that its interior is empty. The latter can be proven using convergence  arguments  similar to those used in the proof of Theorem \ref{theorem1}{, i.e.} every no-signaling behaviour {is} arbitrarily close {(in the weak-convergence sense) to a non-}positive no-signaling beahviou{r.} This may at first blush seem bizarre, but it is actually a typical property of  compact convex sets in infinite-dimensional spaces. Indeed, the sets of probability distributions or quantum sates for infinite-dimensional systems display exactly the same property (see, e.g., Ref. \cite{Haapasalo}).

Lastly, we stress that in the present work we did not touch the question of whether the set $\mathcal M_\mathrm{NS}$ is bounded or not. Doing so would require to introduce an appropriate metric and, as we are dealing with infinite dimensional spaces, the boundedness of the set $\mathcal M_{\mathrm{NS}}$ might depend on its particular choice. For instance, with respect to the Lévy-Prokhorov metric, which is a metric on the set of probability measures  associated to the weak topology, the set of all probability measures is bounded. Hence, for this metric also the no-signaling set is bounded, since the components of behaviours are by definition always probability measures.

\section{Final discussion} 
\label{sec:Conclusion}
We have studied  supraquantum 
Bell correlations in a genuinely CV regime, i.e., without discretisation procedures such as binning \cite{Grangier, Tan, Gilchrist, Munro, SignBinning, RootBinning, LeeJaksch}. To the best of our knowledge, this is the first such investigation reported. Here, genuine CV supraquantumness was witnessed by the violation of the CFRD inequality \cite{CFRD}, which, for the bipartite case, is known not to admit any quantum violation \cite{CFRDSalles,CFRDSallesLong}. We found a class of supraquantum Gaussian PR boxes, whose zero-width limit gives the CV PR boxes. Here, we have explicitly checked the supraquantumness of both Gaussian and CV PR boxes of order $k=2$. Interestingly, due to symmetries in the CFRD inequality, no violation can be found for $k=3$, but supraquantumness of CV PR boxes of higher orders is guaranteed by the supraquantumness of the equivalent boxes in finite dimensions. In turn, the supraquantumness of finite-width Gaussian PR boxes of higher order can be verified violating -- via some appropriate binning -- finite-dimensional Bell inequalities above their quantum limit; but this is outside the scope of this paper.\\
\indent
In addition, we have characterised the set of CV no-signaling correlations from a geometrical point of view. To this end, we devised a mathematical framework to deal with arbitrary CV no-signaling behaviours, based on conditional probability measures instead of conditional probability distributions. 
With this, we have shown that, for CV systems, the convex hull (i.e. the set of all finite convex sums) of all CV PR boxes is dense in the no-signaling set, instead of equal to it as in finite dimensional systems. {In particular, this result tells us that every no-signalling behaviour can be approximated {arbitrary well} by a sequence of behaviours with a finite number of non-zero probability outcomes. Consequently, the nonlocality of every CV no-signaling behaviour can always be detected with discrete Bell inequalities in combination with a binning procedure, for sufficiently large number of bins.
\\
\indent 
Since every CV PR box assigns a non-zero probability to a finite number of outcomes, being thus in one-to-one correspondence with a discrete PR box in the usual finite-dimensional scenario, it is not surprising that every CV PR box is extreme {in the no-signaling set. In contrast, the possibility that all extreme points of the no-signaling set are given by CV PR boxes, as suggested by Conjecture~\ref{conjecture}, appears as more surprising.} 
Indeed, it would evidence a qualitative difference between the structure of quantum theory and that of generic probability theories compatible with the no-signaling principle, a question that has been previously considered in other scenarios too~\cite{Kleinmann}. Namely, in quantum theory we know about the existence of behaviours with an uncountably infinite number of non-zero probability outcomes which are extreme in the set of CV quantum correlations. The latter quantum behaviours can be built, e.g., with extreme quantum POVMs with a continuous spectrum~\cite{Holevo,HeinosaariPellonpaa,Pellonpaa} acting on pure CV entangled states. We leave {the proof (or disproof) of this conjecture as an open question for} future investigations.}
\\
\indent
Another interesting question for future investigations is how to formalize the notion of tightness \cite{Masanes} for CV Bell inequalities; and, in particular, whether the CFRD inequality is tight or not. In finite dimensions, non-trivial tight Bell inequalities without a quantum violation exist in the multipartite scenario \cite{Almeida10}, but no equivalent example is known for bipartite systems. If the CFRD inequality were tight, our results would give it the status of the first known example of a non-trivial tight Bell inequality with no quantum violation in the bipartite setting. \\
\indent
To end up with, far from being just a mere abstract exercise, studying supraquantum non-locality helps us understand quantum non-locality itself.
Efficient tools to study non-locality for discrete systems --such as semi-definite or linear programming-- no longer apply for CV systems; so that the characterization of non-local correlations is a much harder task. 
We thus hope that our findings can be useful for future research, such as, e.g., searching for novel CV Bell inequalities or, more generally, studying generalized-probability theories in CV systems.

\begin{acknowledgments}
LA thanks D. Cavalcanti for useful comments and the International Institute of Physics (IIP) Natal, Brazil, where parts of this work were carried out, for the hospitality. LA's work is financed by the Brazilian ministries MEC and MCTIC and Brazilian agencies CNPq, CAPES, {FAPESP,} FAPERJ, and INCT-IQ. ALF and AK acknowledge CAPES-COFECUB project Ph-855/15 and CNPq for financial support. AK acknowledges financial support from the ERC (Consolidator Grant 683107/TempoQ), and the DFG. A special thank goes to T. Cabana, S. Kaakai and C. Ketterer for the clarification of mathematical details, and to J. P. Pellonp\"a\"a for valuable discussions about extremality problems in quantum mechanics.  
\end{acknowledgments}
\appendix

\section{Radon-Nikodym Theorem}\label{app:Radon}
A measurable space is given by a pair $(X,\Sigma)$, where $X$ is some nonempty set and $\Sigma$ denotes a $\sigma$-algebra on $X$. We define a measure $\nu$ on $X$ to be $\sigma$-finte if $X$ is a countable union of measurable sets $X_i$ with finite measure $\nu(X_i)< \infty$. Note that every probability measure $\mu$ on $\mathbb R$ is also $\sigma$-finite since, on the one hand, $\mathbb R$ can be expressed as countable union of measurable set and, on the other hand, we have by definition that $\mu(A)<1$, for all $A\subset \mathbb R$. Furthermore, a measure $\nu$ is called absolutely continuous with respect to $\mu$, if from $\nu(A)=0$ it follows $\mu(A)=0$, for every measurable set $A\subset X$. Now, we are in the position to state the Radon-Nikodym Theorem.
\begin{theorem}[Radon-Nikodym] 
Given a $\sigma$-finite measure  $\nu$ on $(X,\Sigma)$ that is absolutely continuous with respect to a $\sigma$-finite measure $\mu$ on $(X,\Sigma)$, then it exists a measurable function $f: X \rightarrow [0,\infty) $, referred to as the Radon-Nikodym derivative, such that
\begin{align}
\nu(A) = \int_A f \, d\mu,
\end{align}
for any measurable subset $A \subset X $.
\label{theoremRadon}
\end{theorem}

\section{Proof of Theorem~\ref{theorem1}}\label{app:proofT1}

Before turning to the proof of Theorem~\ref{theorem1} we provide some preliminary notions of the type of convergence that we will use in the following, i.e. the weak convergence{. We s}ay that a sequence of measures $(\mu_n)_{n\in\mathbb N}\in \mathcal{M}_{\Omega}$ converges weakly towards same $\mu \in \mathcal{M}_{\Omega}$, with $n\rightarrow\infty$, if:
\begin{align}
\int_{\Omega}f {d}{\mu}_n \rightarrow \int_{\Omega} f {d}{\mu},
\end{align}
for all $f\in C_b(\Omega)$, where $C_b(\Omega)$ denotes the set of bounded and continuous functions $f:\Omega\rightarrow \mathbb R$. In what follows, if not stated differently, we will always implicitly assume the use of weak convergence for sequences of measures. Moreover, since we often consider behaviours (\textit{i.e.} matrices with entries given by probability measures), we say that a sequence of behaviours $\bm{\mu}_n$ weakly converges to $\bm{\mu}$ if $[\bm{\mu}_n]_{x, y} \rightarrow [\bm{\mu}]_{x, y}, \forall x, y${.

W}eak convergence is a natural choice in the present context because it is directly applicable to sequences of measures without resorting to a specific distributions in terms of some random variables. Other, possibly stronger, notions of convergence do exist but are not required here. Furthermore, as we will see shortly, weak convergence is also meaningful with respect to physical considerations since{,} from experiments{,} one usually extracts some statistical moments of a probability measure instead of the measure itself.

Further on, as stated also in the main text, in order to prove that $\mathrm{Conv}(\mathcal{M}_\mathrm{PR})$ is dense in $\mathcal{M}_\mathrm{NS}$ we need to show that for every behaviour $\boldsymbol\mu \in \mathcal M_\text{NS}$ one can find a sequence in $\mathrm{Conv}(\mathcal M_\text{PR})$ that converges weakly to $\boldsymbol\mu$. In order to keep the proof  of  Theorem~\ref{theorem1} as instructive as possible, we will first 
provide a proof for  the case of behaviours with compact support meaning that their components are probability measures on $\Omega= [-\mathcal K, \mathcal K]^2$.  A generalization of the proof to the most general case $\Omega=\mathbb R\times\mathbb R$ will then be ensued afterwards. 

{Then, the following Lemma holds.}
\begin{lemma}[$\mathrm{Conv}(\mathcal{M}_\mathrm{PR}^{(\mathcal K)})$ dense in $\mathcal{M}_\mathrm{NS}^{(\mathcal K)}$] The closure $\overline{\mathrm{Conv}(\mathcal{M}_\mathrm{PR}^{(\mathcal K)})}$ of $\mathrm{Conv}(\mathcal{M}_\mathrm{PR}^{(\mathcal K)})$ equals $\mathcal{M}_\mathrm{NS}^{(\mathcal K)}$. In other words, $\mathrm{Conv}(\mathcal{M}_\mathrm{PR}^{(\mathcal K)})$ is a \emph{dense subset} of $\mathcal{M}_\mathrm{NS}^{(\mathcal K)}$.
\label{lemma1}
\end{lemma}
Note that, according to the introduced nomenclature in the main text we equiped the corresponding no-signaling set and the set of CV PR boxes with a superscript $\mathcal K$, i.e. $\mathcal{M}^{(\mathcal K)}_\mathrm{NS}$ and $\mathcal{M}^{(\mathcal K)}_\mathrm{PR}$. 
\begin{proof}[Proof of Lemma~\ref{lemma1}]
Without loss of generality we can restrict the following proof to the case $\mathcal K = 1$, i.e. $\Omega = [-1, 1]^2$. The strategy consists {of} explicitly constructing, for every arbitrary $\bm{\mu} \in \mathcal{M}_\mathrm{NS}^{(1)}$, a sequence of behaviours $\bm{\mu}_n\in \mathrm{Conv}(\mathcal{M}_\mathrm{PR}^{(1)})$ that weakly converges to $\bm{\mu}$. 
The proof is divided in three steps: First{, f}or every $\boldsymbol\mu \in \mathcal M_{\mathrm{NS}}${, we} define a sequence of behaviours that weakly converges to $\boldsymbol\mu$. Second, we show that each element of this sequence is indeed a no-signaling behaviour{. T}hird, {we show} that {all such elements can} be expressed as a convex sum of CV PR boxes. 

{For the first step, we} divide the interval $[-1, 1]$ in $n\ge 1$ segments of the same length, denoting each one by $I_n$ ({n}ote that a generalization of the following proof to arbitrary $\mathcal K$'s would simply involve a rescaling of the defined intervals $I_n$.). Next, we define $\bm{\mu}_n$ as follows:
\begin{align}
[\bm{\mu}_n]_{x, y} = \sum_{k, l=1}^n [\bm{\mu}]_{x, y} (I_k \times I_l) \, \delta_{a_k, b_l}, \forall x, y,
\label{eq:def_interval}
\end{align}
where $(a_k, b_l)$ is a point located in the interval $I_k \times I_l$ and the Dirac measure is defined according to Eq.~(\ref{eq:DiracMeasure}). The behaviours $\bm{\mu}_n$ have the same weight as $\bm{\mathrm{\mu}}$ on each of the squares $I_k \times I_l$, but concentrated on a single point $(a_k, b_l)$. In this way, $\bm{\mu}_n$ becomes better and better approximations of $\bm{\mu}$, with increasing $n$. 

To prove that $\bm{\mu}_n$ is indeed weakly converging to $\bm{\mu}$, it suffices to prove that each of its component is weakly converging to the components of $\bm{\mu}$. Let $f$ be a bounded and continuous function defined on the domain $[-1, 1]\times[-1, 1]$. Integrating $f$ with respect to $\mu^{x, y}_n$, yields:
\begin{align}
\int_{[-1, 1]^2} f\mathrm{d} [\bm{\mu}_n]_{x, y} = \sum_{k, l} f(a_k, b_l) [\bm{\mu}]_{x, y} (I_k \times I_l).
\label{eq:Integralmun}
\end{align}
The sum on the right-hand side of Eq.~(\ref{eq:Integralmun}) can be bounded from below and above in the following way:
\begin{align}
\sum_{k, l} [\bm{\mu}]_{x, y} (I_k \times I_l) \underset{I_k \times I_l}{\mathrm{min}} f & \le \int_{[-1,1]\times[-1,1]} f {d} [\bm{\mu}_n]_{x, y} \\
&\le \sum_{k, l} [\bm{\mu}]_{x, y} (I_k \times I_l) \underset{I_k \times I_l}{\mathrm{max}} f,
\end{align}
where $\underset{I_k \times I_l}{\mathrm{min}}$ ($\underset{I_k \times I_l}{\mathrm{max}}$) denotes the minimum (maximum) of the function $f$ over the cell $I_k \times I_l$. The same inequality holds if we integrate $f$ with respect to $[\bm{\mu}]_{x, y}$, and, since $f$ is continuous, this proves that $\int f\mathrm{d} [\bm{\mu}_n]_{x, y} \rightarrow \int f\mathrm{d} [\bm{\mu}]_{x, y}$ and that $[\mu_n]_{x, y} \rightarrow \mu, \forall x, y$. It follows that $\bm{\mu}_n \rightarrow \bm{\mu}$.

{As for the second step, we now } prove that $\bm{\mu}_n$ is no-signaling for all $n$. For a given $n > 0$ and $x, y \in \{0, 1\}^n$, the marginal of $[\mu_n]_{x, y}$ on Bob's side is given by:
\begin{align}
[\boldsymbol{\mu}_n]_{x,y}([-1,1]\times B)  &= \sum_{k, l=1}^n [\bm{\mu}]_{x, y} (I_k \times I_l)  \delta_{a_k, b_l}([-1,1]\times B) \nonumber \\
& = \sum_{l} \delta_{b_l}(B) \sum_{k} [\bm{\mu}]_{x, y} (I_k \times I_l) \nonumber\\
& = \sum_{l} \delta_{b_l}(B) [\bm{\mu}]_{x, y} ([-1,1]\times I_l),
\end{align}
where $\delta_{b_l}$ is the Dirac measure located at $b_l$ in the $l$th interval. Since we know that $\bm{\mathrm{\mu}}$ is a no-signalling behaviour it follows that $[\bm{\mu}]_{x, y} ([-1,1]\times I_l)$ does not depend on $x$ (compare with Eq.~(\ref{eq:constraint1})). The same argument holds for the Alice's marginal and proves that the $\bm{\mu}_n$'s are no-signalling behaviours.

{The third and last step to complete the proof is} to show that $\bm{\mu}_n$ can be written as a convex sum of finitely many CV PR boxes. For this we note that the $\bm{\mu}_n$'s are no-signalling behaviours with a finite number of outcomes (the centers of the intervals $I_{k, l}$) and support $[-1,1]^2$. However, we know form the finite-dimensional case that all behaviours with only finitely many outcomes with non-zero probability can be expressed as a convex combination of finitely many PR boxes. Taking instead their continuous-variable generalizations (\ref{eq:PRBoxesMeasures}), yields the desired decomposition.
\end{proof}

{With Lemma~\ref{lemma1}, we next prove T}heorem~\ref{theorem1}.

\begin{proof}[Proof of Theorem~\ref{theorem1}]
Now, we consider the case $\Omega=\mathbb R\times \mathbb R$. Again, we consider a  $\bm{\mu} \in \mathcal M_\mathrm{PR}$ and want to prove that there exists a sequence $\bm{\mu}_n \in \mathrm{Conv}(\mathcal{M}_\mathrm{PR})$ for which each component converges weakly to the components of $\bm{\mu}$. To do so, we divide $[-n, n]$, with $n \ge 1$, in $2 n^2$ subintervals of length $1/n$ and denote them by $I_n$ as before. Furthermore, we define the components of $\bm{\mu}_n$ as follows:
\begin{align}
[\boldsymbol{\mu}_n]_{x,y} = & \sum_{k, l=1}^{2n^2} [\bm{\mu}]_{x,y} (I_k \times I_l)\delta_{a_k, b_l} + [\boldsymbol\nu_n]_{x,y},
\label{eq:SequenceUnbounded}
\end{align}
where $(a_k, b_l)$ is a point located in the square $I_k \times I_l$, $\delta_{a_k, b_l}$ is the Dirac measure. The first term of Eq.~(\ref{eq:SequenceUnbounded}) corresponds to the same construction as in the compact case treated in Lemma~\ref{lemma1}, whereas the second term $\boldsymbol\nu_n$ is merely necessary to ensure the no-signaling conditions~(\ref{eq:constraint1}) and (\ref{eq:constraint2}) on $\mathbb R\times\mathbb R$. It reads as follows:
\begin{widetext}
\begin{align}
[\boldsymbol{\nu}_n]_{x,y} = &  \sum_{l = -n}^n \Big( [\bm{\mu}]_{x, y} \big(]n, \infty[\, \times\, I_l\big) \delta_{n+1, b_l} + [\bm{\mu}]_{x, y} \big(]-\infty, -n[ \,\times\, I_l\big) \delta_{-(n+1), b_l} \Big) \nonumber \\
& + \sum_{k = -n}^n \Big([\bm{\mu}]_{x, y} \big(I_k\, \times \,]n, \infty[\big) \delta_{a_k, (n+1)} + [\bm{\mu}]_{x, y} \big(I_k\, \times \,]-\infty, -n[\big) \delta_{a_k, -(n+1)} \Big) \nonumber \\
& + [\bm{\mu}]_{x, y} \big(]n, \infty[\, \times\, ]n, \infty[\big) \delta_{n+1, n+1} + [\bm{\mu}]_{x, y} \big(]-\infty, -n[\, \times\, ]n, \infty[\big) \delta_{-(n+1), n+1} \nonumber \\
& + [\bm{\mu}]_{x, y} \big(]n, \infty[ \times ]-\infty, -n[\big) \delta_{n+1, -(n+1)} + [\bm{\mu}]_{x, y} \big(]-\infty, -n[\, \times\, ]-\infty, -n[\big) \delta_{-(n+1), -(n+1)},
\label{eq:SequenceUnboundedExtra}
\end{align}
\end{widetext}
{where $]a,b[$ refers to an open interval bounded by $a$ and $b$, respectively.}
 Note that, in contrast to the compact case treated in Lemma~\ref{lemma1},  the measures ${[\boldsymbol\mu_n]_{x,y}}$ are defined on different intervals for different $n$'s. We will now complete the proof of Theorem~\ref{theorem1} by showing the weak convergence of this sequence in the general case. The other parts of the proof remain unchanged. 

Let $f \in \mathcal{C}_b(\mathbb{R^2})$ and $\epsilon \in [0, 1]$, we want to prove that there exists an $n_0\in\mathbb N$ such that $|\int_{\mathbb{R}^2}f {d} \bm{\mu}_n - \int_{\mathbb{R}^2}f {d} \boldsymbol\mu| < \epsilon$ for all $n > n_0$, where this inequality should be understood as component wise inequality. Since $\bm{\mu}$ is a set of probability measures and $f$ is a bounded function, there exists an $n_1\in \mathbb N$ such that:
\begin{align}
[\bm{\mu}]_{x, y} (\mathbb{R}^2 \setminus [-n_1, n_1]) < \mathrm{min} \left(\epsilon, \frac{\epsilon}{\mathrm{max}_{\mathbb{R}^2} |f|} \right),
\end{align}
for all $(x, y)$ and all $n > n_1$. It follows that:
\begin{align}
\Big| & \int_{\mathbb{R}^2}f {d}[\boldsymbol{\mu}_n]_{x,y} - \int_{\mathbb{R}^2}f {d} [\bm{\mu}]_{x, y} \Big| \nonumber \\
&< \Big|\int_{\mathbb{R}^2\setminus [-n_1, n_1]^2}f {d}[\boldsymbol{\mu}_n]_{x,y} - \int_{\mathbb{R}^2\setminus [-n_1, n_1]^2}f {d} [\bm{\mu}]_{x, y} \Big| \nonumber \\
&+ \Big|\int_{[-n_1, n_1]^2}f {d}[\boldsymbol{\mu}_n]_{x,y} - \int_{[-n_1, n_1]^2}f {d} [\bm{\mu}]_{x, y} \Big|.
\label{eq:proofInequality1}
\end{align}
While the first term on the right and side of inequality (\ref{eq:proofInequality1}) becomes:
\begin{align}
\Big|&\int_{\mathbb{R}^2\setminus [-n_1, n_1]^2}f {d}[\boldsymbol{\mu}_n]_{x,y} - \int_{\mathbb{R}^2\setminus [-n_1, n_1]^2}f {d}[\bm{\mu}]_{x, y} \Big| \nonumber \\ 
& \le \Big|\int_{\mathbb{R}^2\setminus [-n_1, n_1]^2}f {d}[\boldsymbol{\mu}_n]_{x,y}\Big| + \Big|\int_{\mathbb{R}^2\setminus [-n_1, n_1]^2}f {d} [\bm{\mu}]_{x, y} \Big| \nonumber \\
& \le \mathrm{max}_{\mathbb{R}^2} |f| [\boldsymbol{\mu}_n]_{x,y}(\mathbb{R}^2 \setminus [-n_1, n_1]^2) + \epsilon \nonumber \\
& = \mathrm{max}_{\mathbb{R}^2} |f| [\bm{\mu}]_{x, y} (\mathbb{R}^2 \setminus [-n_1, n_1]^2) + \epsilon \nonumber\\
& \le 2 \epsilon,
\label{eq:proofInequality2}
\end{align}
the second term contains an integration over a compact area{,} which allows us to use the statement of Lemma~\ref{lemma1}. Hence, we can conclude that this term is smaller than $\epsilon$ for sufficiently large $n$. Note that Lemma~\ref{lemma1} does not  apply directly here since the considered sequence of behaviours is not no-signaling on the compact domain $[-n_1, n_1]^2$, but rather on $\mathbb R^2$. However, dropping the no-signaling condition does not contradict with the convergence of this sequence.  By combining inequalities (\ref{eq:proofInequality1}) and (\ref{eq:proofInequality2}) we finally arrive at
\begin{align}
\Big|\int_{\mathbb{R}^2}f {d} [\bm{\mu}_n]_{x, y} - \int_{\mathbb{R}^2}f {d} [\bm{\mu}]_{x, y}\Big| < 3 \epsilon,
\end{align}
for $n$ sufficiently large. This quantity goes to zero as $\epsilon$ goes to zero and thus $\bm{\mu}_n$ weakly converges to $\bm{\mu}$.
\end{proof}

\section{Concerning Conjecture~\ref{conjecture}}\label{app:limit_k}
Here we construct a specific example of a sequence of CV PR boxes, with increasing order $k$, whose limit is not an extreme no-signaling behavior anymore. This suggests that one cannot obtain extreme no-signaling behaviors as limits of a sequences of CV PR boxes when the order $k$ goes to infinity. We will restrict ourselves to measures on $[0, 1]^2$ but it can be straightforwardly extended to $\mathbb{R}^2$.
\begin{proof}
We prove that there is a sequence $\bm{\mu}_n \in \mathcal{M}^{(1)}_{\mathrm{PR}}$ that converges to an element $\bm{\mu}$ that is outside of $\mathcal{M}^{(1)}_{\mathrm{PR}}$. Let $\bm{\mu}$ be the set of measures where the two outcomes are always perfectly correlated for all settings: $\mu^{x, y}(a, b) = \delta(a - b)$. $\bm{\mu}$ is clearly no-signaling, but not extreme.

We define $\bm{\mu}_n$ as follows:
\begin{align}
\mu_n^{x, y} =\begin{cases}
\frac{1}{n} \sum_{k = 0}^n \delta_{\frac{k}{n}, \frac{k}{n}}, &\mathrm{for}~x\cdot y = 0, \\
 \frac{1}{n} \left[\sum_{k = 0}^{n - 1} \delta_{\frac{k}{n}, \frac{k+1}{n}} + \delta_{1, 0} \right],  &\mathrm{for}~x\cdot y = 1,
\end{cases}
\end{align}
which yields
\begin{align}
&\iint_{[0, 1]^2} f(a, b) \mu^{x, y}_n(a, b)\\ 
&=\begin{cases} \frac{1}{n} \left[\sum_{k = 0}^{n} f\left(\frac{k}{n}, \frac{k}{n}\right) \right],  &\mathrm{for}~x\cdot y = 0, \\
 \frac{1}{n} \left[\sum_{k = 0}^{n - 1} f\left(\frac{k}{n}, \frac{k+1}{n}\right) + f(1, 0) \right],  &\mathrm{for}~x\cdot y = 1,
\end{cases}
\end{align}
where $f\in C_b([0, 1]^2)$. Now, by a applying standard integration theory it follows that $\left[\frac{1}{n}\sum_{k = 0}^{n - 1} f\left(\frac{k}{n}, \frac{k+1}{n}\right) + f(1, 0) \right] \rightarrow \int_{[0, 1]} f(a, a) = \iint_{[0, 1]^2} f(a, b) \mu^{x, y}(a, b)$. We thus proved that $\bm{\mu}_n$ converges to an element that is outside of $\mathcal{M}^{(1)}_{\mathrm{PR}}$ (since $\bm{\mu}$ has an infinite number of outcomes contrary to all elements of $\mathcal{M}^{(1)}_{\mathrm{PR}}$).
\end{proof}


\end{document}